\documentclass[12pt,reqno]{amsart}%
\usepackage{amsmath,amsthm,indentfirst}
\usepackage{amsfonts}
\usepackage{amssymb}
\usepackage{amsbsy,bm}
\usepackage{amsaddr}

\usepackage{graphicx}

\renewcommand{\mathbf}{\bm}
\theoremstyle{plain}
\newtheorem{theorem}{Theorem}

\newtheorem{proposition}[theorem]{Proposition}

\theoremstyle{definition}
\newtheorem{definition}[theorem]{Definition}

\newtheorem{corollary}[theorem]{Corollary}

\theoremstyle{remark}
\newtheorem{remark}[theorem]{Remark}

\numberwithin{equation}{subsection}
\numberwithin{theorem}{subsection}
\normalfont\upshape
\AtBeginDocument{\addtocontents{toc}{\vspace{-1em}}}

\title{\textbf{Double-Graded Supersymmetric Quantum Mechanics}}
 \author{\scshape\textbf{ A\lowercase{ndrew} J\lowercase{ames} B\lowercase{ruce}}}
   \address{Mathematics Research Unit, University of Luxembourg,\\ Maison du Nombre 6, avenue de la Fonte,
L-4364 Esch-sur-Alzette}
   \email{andrewjamesbruce@googlemail.com}
\author{\scshape\textbf{ S\lowercase{teven} D\lowercase{uplij}}}
   \address{ 
Universit\"at M\"unster,
D-48149 M\"unster,
Deutschland}
   \email{douplii@uni-muenster.de}

\date{April 12, 2019}
\usepackage[usenames,dvipsnames]{xcolor}
\usepackage{color}

\begin{document}

\renewcommand{\leftmark}{\scshape\textbf{ A\lowercase{ndrew} J\lowercase{ames} B\lowercase{ruce}} and  \textbf{ S\lowercase{teven} D\lowercase{uplij}}}
\thispagestyle{empty}
\begin{abstract}
A quantum mechanical model that realizes the $ \mathbb{Z}_2 \times \mathbb{Z}_2$-graded generalization of the one-dimensional supertranslation algebra is proposed. This model shares some features with the well-known Witten model and is related to parasupersymmetric quantum mechanics, though the model is not directly equivalent to either of these. The purpose of this paper is to show that novel ``higher gradings'' are possible in the context of non-relativistic quantum mechanics.
\end{abstract}

\maketitle
\tableofcontents
\thispagestyle{empty}
\newpage

\mbox{}
\vskip 1cm

\section{\scshape Introduction}

Supersymmetry in the context of string theory and quantum field theory has a long history dating back to the 1970s \cite{gol/lik,vol/aku,wes/zum2} (for the history
and development of supersymmetry, see, e.g., \cite{dup/sie/bag,kan/shi}). However, we know that if Nature does indeed utilize supersymmetry then supersymmetry must be broken. We have in mind here models of high energy physics. Superymmetry can also be realized in condensed matter physics (see \cite{efetov} and later works). Understanding the possible mechanisms for breaking supersymmetry is vital to construct realistic theories.

In 1981, Witten \cite{Wit} introduced (non-relativistic) supersymmetric quantum mechanics as a toy model to discuss supersymmetry breaking via instanton effects. This model is the simplest non-trivial model that exhibits the general features of a supersymmetric  field theory. Since then the subject has grown and many useful tools in quantum mechanics have been developed with supersymmetric quantum mechanics as their root.

In particular, Witten's model provides an alternative scheme to the factorization method;  gives rise to an understanding of exactly solvable potentials via shape invariance; has applications in the inverse scattering method, etc. (see, e.g., \cite{Bag,Coo/Kha/Uda,Suk,Szc}).  For the most clear exposition of
Witten's theory, see \cite{fub/rab}, followed by the many papers on ``form-invariance''.

Since Witten's original work \cite{Wit}, there has been several generalizations of his quantum mechanical model including parasupersymmetric quantum mechanics \cite{Rub/Spi}, orthosupersymmetric quantum mechanics \cite{Kha/Mis/Raj}, and fractional supersymmetric quantum mechanics \cite{Dur}, all of which are motivated by generalizations of standard statistics. Independently, the use of para-Grassmann variables in mechanics has been implemented by Gershun and Tkach \cite{ger/tka84,ger/tka85}.

Here we examine another generalization in which the graded structure of the theory is different from the standard supersymmetric quantum mechanics. In particular, we employ a ``double-grading'' using $\mathbb{Z}_2^2 := \mathbb{Z}_2 \times \mathbb{Z}_2$ rather than just the standard grading via $\mathbb{Z}_2$.   \par
Recently, in \cite{Bru} it was proposed a generalization of the $\mathcal{N}$-extended supersymmetry algebra to the setting of $\mathbb{Z}_2^n$-Lie algebras introduced in \cite{mol,mol2010}. The main aim of \cite{Bru} was to establish a geometric understanding of the algebra in terms of $\mathbb{Z}_2^n$-manifolds (see \cite{Cov/Gra/Pon1,Cov/Gra/Pon2}), i.e., to build a generalization of super-Minkowski space-time  (see, also \cite{tol14a}). Loosely, $\mathbb{Z}_2^n$-manifolds are `manifolds' for which the structure sheaf has a $\mathbb{Z}_2^n$-grading and the commutation rules for the local coordinates come from the standard scalar product. The case of $n=1$ is just the theory of standard supermanifolds
\cite{berezin,lei1,leites}.  However, in \cite{Bru} no examples of classical or quantum systems that  have the proposed $\mathbb{Z}_2^n$-Lie algebra as a symmetry were given. We rectify this omission here. \par
We present a \emph{double-graded}  supersymmetric quantum mechanical model on the real line that exhibits the required $\mathbb{Z}_2^n$-symmetry for the case of $n=2$. This model resembles Witten's version of supersymmetric quantum mechanics \cite{Wit} ( also see \cite{aku/pas,gen/kri}), as well as various models with extended supersymmetry \cite{ger/tka79,pas1986,aku/kud,fub/rab,pas/top}. However, the  differences are as follows:
\begin{itemize}
\item the underlying algebra will be a $\mathbb{Z}_2^2$-Lie algebra and not a super Lie algebra (i.e., a $\mathbb{Z}_2$-Lie algebra),
\item the Hilbert space will be $\mathcal{H} :=  L^2(\mathbb{R}) \otimes \mathbb{C}^4$ in order to take account of the four possible `spin' states, and
\item we will allow for a central charge.
\end{itemize}
\par

We remark that $\mathbb{Z}_2^2$-gradings appear in several guises within mathematical physics, for example in relation to the symmetries of the L\'{e}vy-Leblond equation \cite{Aiz/Kuz/Tan/Top} and parastatistics \cite{Gre,Vol} (see, also \cite{tol14}). We will show how the models presented here are related to, but not equivalent to models that posses parasupersymmetries \cite{Rub/Spi}  (see \cite{Bec/Deb1,Bec/Deb2,Kha1,ger/tka84}, on a \textit{filtered} generalization of Lie superalgebras, see \cite{lei/ser90,lei/ser01,bar95,vol85}).\par
The reader should observe that Scheunert proved a theorem reducing ``colored'' Lie algebras to either Lie algebras or Lie superalgebras \cite{Sch}, while Neklyudova proved an analogue of this theorem for ``colored'' graded-commutative and associative algebras \cite{leites13}. However, these theorems deal only with algebras, not with pairs (algebra, module over it), and for such pairs the direct analogues of the Neklyudova-Scheunert theorem do not hold. This fact is an indicator that our construction, as well as \cite{Cov/Gra/Pon1,Cov/Gra/Pon2} (and the earlier works \cite{cov/ovs/pon,mor/ovs10}), are not meaningless. \par
In Section \ref{sec:Prelim} we present the $1$-dimensional  $\mathbb{Z}_2^2$-supertranslation algebra, we define quantum mechanical systems that are  $\mathbb{Z}_2^2$-supersymmetric and derive some direct consequences of this definition. We will assume the reader is already somewhat familiar with Witten's model.   In Section \ref{sec:ZSUSYQM} we define and examine a specific model that is akin to Witten's model.  Some closing remarks are in Section \ref{sec:Concl}.

\section{\scshape Preliminaries}\label{sec:Prelim}
\subsection{$\mathbb{Z}_2^2$-Lie Algebras}
Let us recall the notion of a $\mathbb{Z}_2^2$-Lie algebra (see \cite{rit/wyl,Sch}). The extension to $\mathbb{Z}_2^n$-Lie algebras is straightforward, though in this note we will not consider the general case. A $\mathbb{Z}_2^2$-graded vector space is a vector space (over $\mathbb{R}$ or $\mathbb{C}$) that is the direct sum of homogeneous vector spaces
$$\mathfrak{g} = \mathfrak{g}_{00} \oplus \mathfrak{g}_{01} \oplus \mathfrak{g}_{11} \oplus \mathfrak{g}_{10}\,.$$
Note that we have chosen a particular ordering for the elements of $\mathbb{Z}_2^2 := \mathbb{Z}_2 \times \mathbb{Z}_2$. This ordering will be convenient for our applications in quantum mechanics.  We will denote the $\mathbb{Z}_2^2$-degree of  an element of $\mathfrak{g}$, $a$ (say), as $\mathrm{deg}(a) \in \mathbb{Z}_2^2$.   We define the \emph{even subspace} of $\mathfrak{g}$ to be $ \mathfrak{g}_{00} \oplus \mathfrak{g}_{11}$ and the \emph{odd subspace} of $\mathfrak{g}$ to be $ \mathfrak{g}_{01} \oplus  \mathfrak{g}_{10}$. That is we pass to a $\mathbb{Z}_2$-grading via the total degree, i.e., the sum of the components of the $\mathbb{Z}_2^2$-degree. We will denote the standard scalar product on $\mathbb{Z}_2^2$ by $\langle -, - \rangle$. That is, if $\mathrm{deg}(a) = (\gamma_1, \gamma_2)$ and $\mathrm{deg}(b) = (\gamma^\prime_1, \gamma^\prime_2)$, then $\langle \mathrm{deg}(a), \mathrm{deg}(b) \rangle = \gamma_1 \gamma^\prime_1 +  \gamma_2 \gamma^\prime_2$. \par
A \emph{$\mathbb{Z}_2^2$-Lie algebra} is a  $\mathbb{Z}_2^2$-graded vector space equipped with a bi-linear operation, $[-,-]$, such that for homogeneous elements  $a,b$ and $c\in \mathfrak{g}$, the following are satisfied:
\begin{enumerate}
\item $\mathrm{deg}([a,b]) = \mathrm{deg}(a) + \mathrm{deg}(b)$,
\item $[a,b] = {-} (-1)^{\langle \mathrm{deg}(a), \mathrm{deg}(b) \rangle} [b,a]$,
\item $[a,[b,c] ] = [[a,b],c] = + (-1)^{\langle \mathrm{deg}(a), \mathrm{deg}(b) \rangle} [b, [a,c]]$.
\end{enumerate}
Extension to inhomogeneous elements is via linearity.
\begin{remark}
We have written the Jacobi identity for a $\mathbb{Z}_2^2$-Lie algebra in Loday-Leibniz form (see \cite{Lod}). Note that this form has a direct interpretation independently of the symmetry of the bracket. In particular, the notion of a
$\mathbb{Z}_2^2$-Loday-Leibniz algebra is clear, i.e., just drop the symmetry condition.
\end{remark}
\subsection{The $\mathbb{Z}_2^2$-supertranslation algebra in one dimension}
The starting place is the following $\mathbb{Z}_2^2$-graded Lie algebra (see \cite{rit/wyl,Sch}) with generators $H_{00}, Q_{01}, Q_{10}$ and $Z_{11}$ of $\mathbb{Z}_2^2$-degrees $(0,0)$, $(0,1)$, $(1,0)$  and $(1,1)$, respectively, given by
\begin{align}\label{eqn:Z22SuperAlg}
& [Q_{01},Q_{01}] = [Q_{10}, Q_{10}] = \frac{1}{2}H_{00}, && [Q_{10},Q_{01}] = \frac{1}{2}Z_{11}\,,
\end{align}
where all other Lie brackets vanish - this belongs to the $C(2,s)$-family of generalized Lie algebras following \cite{rit/wyl}. The above algebra was proposed by Bruce \cite{Bru} and will be referred to as the \emph{$\mathbb{Z}_2^2$-supertranslation algebra (in one dimension)}. Up to conventions, this is essentially two copies of supersymmetry in one-dimension and a central extension. However, one must take care as we have $\mathbb{Z}_2^2$-graded Lie brackets, i.e., we employ a novel double-grading. This algebra should be compared with the extended SUSY algebra with central charges given in \cite{fau/spe,fuj/has/nis,Nie/Nik}.\par
Let us assume that we can find a representation of this $\mathbb{Z}_2^2$-Lie algebra as operators on some Hilbert space -- the space of states of some quantum system -- with respect to the $\mathbb{Z}_2^2$-graded commutator.  Explicitly, if we have two homogeneous operators, $X$ and $Y$, then the commutator is defined as
\begin{equation}\label{eqn:Z22Com}
[X,Y] := X \circ Y {-}(-1)^{\langle \textnormal{deg}(X),\textnormal{deg}(Y) \rangle } \: Y \circ X.
\end{equation}
It is easy to verify that we the $\mathbb{Z}_2^2$-graded commutator satisfies the requirements to define a $\mathbb{Z}_2^2$-Lie algebra. Extension to inhomogeneous operators is via linearity.  Furthermore, we will take the $\mathbb{Z}_2^2$-\emph{supercharges} $Q_{01}$ and $Q_{10}$ to be Hermitian. Thus, the operator $H_{00}$ is also Hermitian, naturally, this will have the interpretation as a \emph{Hamiltonian}, while the \emph{central charge} $Z_{11}$ will be anti-Hermitian.  From these considerations, we are led to the following definition.
\begin{definition}\label{def:ZQMSystem}
A \emph{$\mathbb{Z}_2^2$-supersymmetric quantum mechanical system}  (\emph{$\mathbb{Z}_2^2$-SUSY QM system}) or a \emph{double-graded supersymmetric quantum mechanical system} is the quintuple
$$\big \{\mathcal{H}; H_{00}, Q_{01}, Q_{10}, Z_{11}   \big\},$$
where
\begin{enumerate}
\item $\mathcal{H} = \mathcal{H}_{00} \oplus \mathcal{H}_{01} \oplus \mathcal{H}_{11} \oplus \mathcal{H}_{10}$, is a $\mathbb{Z}_2^2$-graded Hilbert space, and
\item  $H_{00}, Q_{01}$ and $Q_{10}$ are Hermitian operators and  $Z_{11}$ is an anti-Hermitian operator all of the indicated $\mathbb{Z}_2^2$-degrees, that satisfy the  $\mathbb{Z}_2^2$-supertranslation algebra \eqref{eqn:Z22SuperAlg} with respect to the $\mathbb{Z}_2^2$-graded commutator \eqref{eqn:Z22Com}.
\end{enumerate}
\end{definition}
Directly from the definition of a double-graded supersymmetric quantum mechanical system we can derive some results that apply for all such systems.
\begin{theorem}
For any $\mathbb{Z}_2^2$-SUSY QM system the expectation value of the Hamiltonian operator $H_{00}$ is positive.
\end{theorem}
\begin{proof}
This follows in exactly the same way as standard supersymmetric quantum mechanics. In particular, as we have Hermitian charges we can write $4 Q_{\mathbf{i}}^\dag Q_{\mathbf{i}}  = H_{00}$, here $\mathbf{i} = 01$ or $10$. Then for any state
$$\langle \psi | \frac{1}{4} H_{00} |\psi \rangle = \langle \psi | Q_{\mathbf{i}}^\dag Q_{\mathbf{i}} |\psi \rangle  = \parallel Q_{\mathbf{i}} |\psi\rangle   \parallel^2  ~\geq 0 \,.$$
\end{proof}
That is, just as in standard supersymmetric quantum mechanics the energy is always positive for \underline{any} $\mathbb{Z}_2^2$-SUSY QM system. \par
By definition a $\mathbb{Z}_2^2$-SUSY QM system has good  $\mathbb{Z}_2^2$-supersymmetry  if and only if there exists at least one state that is annihilated by both of the $\mathbb{Z}^2_2$-supercharges $Q_{01}$ and $Q_{10}$. Clearly, this implies that such states have zero energy as $H_{00} \sim Q_{01}^2 = Q_{10}^2$. Otherwise, we say that $\mathbb{Z}_2^2$-supersymmetry is broken. Moreover, it is clear from the fact that $Z_{11} \sim [Q_{01}, Q_{01}]$ that zero energy states are also annihilated by the central charge. We have thus proved the following:
\begin{proposition}\
\begin{enumerate}
\item A $\mathbb{Z}_2^2$-SUSY QM system has good $\mathbb{Z}_2^2$-supersymmetry if and only if the vacuum state has zero energy, otherwise the energy is strictly positive, and
\item if a $\mathbb{Z}_2^2$-SUSY QM system has good $\mathbb{Z}_2^2$-supersymmetry, then the vacuum state is annihilated by the central charge.
\end{enumerate}
\end{proposition}
\subsection{The relation with parasupersymmetry}
In the context of parasupersymmetric quantum mechanics  (see \cite{Bec/Deb1,Rub/Spi}), Beckers and Debergh \cite{Bec/Deb2} define, via Green's ansatz (see \cite{Gre}), the  following  algebra (here $\{- ,-\}$ is the usual anticommutator and $[-,-]$ the usual commutator)
\begin{align}\label{eq:BDalg}
& \{Q_i, Q_i \} = \{Q_i^\dag,  Q_i^\dag \} =0 , && \{Q_i,  Q_i^\dag \} =H, \\
\nonumber & [Q_i, Q_j] = [Q_i ,Q_j^\dag] = [Q_i^\dag,  Q_j^\dag] =0 ~ (i\neq j), && [H, Q_i]=0,
\end{align}
for $i,j = 1,2$.
\begin{proposition}
There exists a non-canonical morphism from the   Beckers--Debergh algebra \eqref{eq:BDalg} to the $\mathbb{Z}_2^2$-supertranslation algebra \eqref{eqn:Z22SuperAlg} with $Z_{11} =0$.
\end{proposition}
\begin{proof}
We define $Q_{01} :=  \frac{1}{2}(Q_1 + Q_1^\dag)$, $Q_{10} :=  \frac{1}{2}(Q_2 + Q_2^\dag)$ and $H_{00} := H$. Note that this is a choice and we could equally have made the other obvious choice here with  $Q_{01}$ and $Q_{10}$.  Direct calculation shows that we recover the algebra \eqref{eqn:Z22SuperAlg} with $Z=0$:
\begin{align*}
& [Q_{01} ,Q_{01}] = \frac{1}{2}(Q_1 + Q_1^\dag) (Q_1 + Q_1^\dag) = \frac{1}{2}\{Q_1, Q_1^\dag \} = \frac{1}{2}H_{00},\\
& [Q_{10} ,Q_{10}] = \frac{1}{2}(Q_2 + Q_2^\dag) (Q_2 + Q_2^\dag) = \frac{1}{2}\{Q_2, Q_2^\dag \} = \frac{1}{2}H_{00},\\
&[Q_{10}, Q_{01}] = \frac{1}{4} \big([Q_2, Q_1] + [Q_2,Q_1^\dag] + [Q_2^\dag, Q_1] + [Q_2^\dag, Q_1^\dag]  \big) = \frac{1}{2}Z_{11}=0,\\
& [Q_{01}, H_{00}] = \frac{1}{2}\big( [Q_1, H] + [Q_1^\dag ,H]\big) =0,\\
& [Q_{10}, H_{00}] = \frac{1}{2}\big( [Q_2, H] + [Q_2^\dag ,H]\big) =0,\\
\end{align*}
where we have directly used the algebra \eqref{eq:BDalg} and the fact that $H$ is Hermitian.
\end{proof}
We stress the point that this morphism is not invertible - there is no way to uniquely decompose a Hermitian operator into the sum of a non-Hermitian operator and its Hermitian conjugate. Also we draw attention to the fact that the central charge has to vanish in order to construct the above morphism. Thus, clearly, we do not have an isomorphism between the Beckers--Debergh algebra and the $\mathbb{Z}_2^2$-supertranslation algebra. In other words, although there are clear similarities between the $\mathbb{Z}_2^2$-supertranslation algebra and parastatistics/parasupersymmetry, the two concepts are not the same.
\subsection{`Higher' Pauli matrices}
The $\mathbb{Z}_2^2$-degrees of freedom for a given quantum mechanical are given by vectors in $\mathbb{C}^4$ and not just $\mathbb{C}^2$  (i.e., we have more than just spin  ``up'' and ``down''). Thus, we cannot directly use the Pauli matrices. Instead, we have to use the so-called ``Sigma'' and ``alpha'' matrices which are built by placing the Pauli matrices (and the identity matrix) on the diagonal and anti-diagonal, respectively, i.e., we use the direct sum and the skew sum of matrices. That is, we define $\Sigma_i := \sigma_i \oplus \sigma_i$ and $\alpha_i := \sigma_i \ominus \sigma_i$ ($i = 0,1,2,3$) where $\sigma_i$ are the  Pauli matrices $(\sigma_0 := \textnormal{id}_{2\times 2})$. Explicitly we have the following matrices:
\[ \Sigma_0 = \left( \begin{array}{rrrr}
1 & 0 & 0 & 0 \\
0 & 1 & 0 & 0  \\
0 & 0 & 1 & 0\\
0 & 0 & 0 & 1\\
\end{array} \right), \hspace{25pt}
 \Sigma_1 = \left( \begin{array}{rrrr}
0 & 1 & 0 & 0 \\
1 & 0 & 0 & 0  \\
0 & 0 & 0  & 1\\
0 & 0 & 1  & 0\\
\end{array} \right), \] \\
\[
\Sigma_2 = \left( \begin{array}{rrrr}
0 & - \mathrm{i} & 0 & 0 \\
\mathrm{i} & 0 & 0 & 0  \\
0 & 0 & 0  & -  \mathrm{i}\\
0 & 0 & \mathrm{i}  & 0\\
\end{array} \right), \hspace{25pt}
 \Sigma_3 = \left( \begin{array}{rrrr}
1 & 0 & 0 & 0 \\
0 & -1 & 0 & 0  \\
0 & 0 & 1 & 0\\
0 & 0 & 0 & -1\\
\end{array} \right),\]
and
\[ \alpha_0 = \left( \begin{array}{rrrr}
0 & 0 & 1 & 0 \\
0 & 0 & 0 & 1  \\
1 & 0 & 0 & 0\\
0 & 1 & 0 & 0\\
\end{array} \right), \hspace{25pt}
 \alpha_1 = \left( \begin{array}{rrrr}
0 & 0 & 0 & 1 \\
0 & 0 & 1 & 0  \\
0 & 1 & 0  & 0\\
1 & 0 & 0  & 0\\
\end{array} \right), \] \\
\[
\alpha_2 = \left( \begin{array}{rrrr}
0 & 0 & 0 & - \mathrm{i} \\
0 & 0 & \mathrm{i} & 0  \\
0 & - \mathrm{i} & 0  & 0\\
\mathrm{i} & 0 & 0  & 0\\
\end{array} \right), \hspace{25pt}
 \alpha_3 = \left( \begin{array}{rrrr}
0 & 0 & 1 & 0 \\
0 & 0 & 0 & -1  \\
1 & 0 & 0 & 0\\
0 &-1 & 0 & 0\\
\end{array} \right).\]

\bigskip

The algebraic properties of these matrices follow from that of the Pauli matrices (see for example \cite[pages 209--212]{Arf/Web}).  In particular, if we define $[A,B]_\mp = AB \mp BA$ and follow the definitions though, we obtain
\begin{align}
& [\Sigma_i, \Sigma_j]_\mp = [\alpha_i, \alpha_j]_\mp = [\sigma_i, \sigma_j]_\mp \oplus [\sigma_i, \sigma_j]_\mp,\\
& [\Sigma_i, \alpha_j]_\mp = [\sigma_i, \sigma_j]_\mp \ominus [\sigma_i, \sigma_j]_\mp\,.
\end{align}
Thus, we see that the (anti-)commutators of the  $\Sigma$ and $\alpha$ matrices are completely determined by the (anti-)commutators of the Pauli matrices.\par
We consider the $\Sigma$ and $\alpha$ matrices to carry $\mathbb{Z}_2^2$-degree as defined by their action on $\mathbb{C}^4$. We use the following decomposition:
\begin{equation}
\mathbb{C}^4 = \mathbb{C}_{00} \oplus \mathbb{C}_{01} \oplus \mathbb{C}_{11} \oplus \mathbb{C}_{10}.
\end{equation}
We can then assign the following degrees:
\begin{align}\label{eqn:DegMat}
& \mathrm{deg}(\Sigma_0) = (0,0), & \mathrm{deg}(\Sigma_1) = (0,1), & & \mathrm{deg}(\Sigma_2) = (0,1),&&  \mathrm{deg}(\Sigma_3) = (0,0),\\
\nonumber & \mathrm{deg}(\alpha_0) = (1,1), & \mathrm{deg}(\alpha_1) = (1,0), & & \mathrm{deg}(\alpha_2) = (1,0),&&  \mathrm{deg}(\alpha_3) = (1,1).
\end{align}
This assignment of the $\mathbb{Z}_2^n$-degree will be essential in how we define the charges in our  model.\\

\noindent \textbf{Warning.} From now on all commutators will be $\mathbb{Z}_2^2$-graded commutators unless otherwise stated (see \eqref{eqn:Z22Com}).
\begin{proposition}
With the above $\mathbb{Z}_2^2$-grading (see \eqref{eqn:DegMat}), the vector space (over $\mathbb{C}$) spanned by the $\Sigma$ and $\alpha$ matrices forms a $\mathbb{Z}_2^2$-Lie algebra with respect to the commutator \eqref{eqn:Z22Com}.
\end{proposition}
\begin{proof}
First, we observe that the $\Sigma$ and $\alpha$ matrices are linearly independent. Secondly, it is clear that we have closure under the ($\mathbb{Z}_2^2$-graded) commutator. This follows as all the properties of the $\Sigma$ and $\alpha$ matrices are inherited from the Pauli matrices and that the Pauli matrices are closed under both (non-graded) (anti-)commutators. Thirdly,  the  ($\mathbb{Z}_2^2$-graded) Jacobi identity is obviously satisfied as we are dealing with the commutator.
\end{proof}
Explicitly, the non-vanishing brackets are:
\begin{align}
& [\Sigma_1 ,\Sigma_1] = 2 \: \Sigma_0, && [\Sigma_1 ,\Sigma_3] = - 2 \mathrm{i} \: \Sigma_2, && [\Sigma_2 ,\Sigma_2] = 2 \: \Sigma_0, && [\Sigma_2 ,\Sigma_3] = 2 \mathrm{i} \: \Sigma_1,\\
\nonumber
& [\alpha_0 ,\alpha_1] = 2 \: \Sigma_1, && [\alpha_0 ,\alpha_2] =  2  \: \Sigma_2, && [\alpha_1 ,\alpha_1] = 2 \: \Sigma_0, && [\alpha_2 ,\alpha_2] = 2 \: \Sigma_0,\\
\nonumber
& [\Sigma_1 ,\alpha_0] = 2 \: \alpha_1, && [\Sigma_1 ,\alpha_2] =  2 \mathrm{i} \: \alpha_3, && [\Sigma_2 ,\alpha_0] = 2 \: \alpha_2, && [\Sigma_2 ,\alpha_1] = -2 \mathrm{i} \: \alpha_3.
\end{align}
%
\section{\scshape A double-graded supersymmetric quantum mechanical system}\label{sec:ZSUSYQM}

\subsection{The Specific Model}
 In order to build a specific double-graded supersymmetric quantum mechanical system (see Definition \ref{def:ZQMSystem}) we propose the following operators acting on $\mathcal{H} :=  L^2(\mathbb{R}) \otimes \mathbb{C}^4$:
\begin{align}\label{eqn:Model}
& Q_{01} = \frac{1}{2} \left( \frac{p}{\sqrt{2 m}} \otimes \Sigma_1 + W \otimes \Sigma_2\right), \;\;\; Q_{10} = \frac{1}{2} \left( \frac{p}{\sqrt{2 m}} \otimes \alpha_2  -  W \otimes \alpha_1\right),\nonumber\\
\nonumber & H_{00} = \left(\frac{p^2}{2m} + W^2 \right) \otimes \Sigma_0 + \frac{\hbar}{ \sqrt{2 m}} W^\prime \otimes \Sigma_3,\\&  Z_{11} =  - \mathrm{i} \left( \left(\frac{p^2}{2m} + W^2 \right) \otimes \alpha_3 + \frac{\hbar}{ \sqrt{2 m}} W^\prime \otimes \alpha_0 \right).
\end{align}
Here $W := W(x) \in C^\infty(\mathbb{R})$ and $W^\prime :=  \frac{\mathrm{d}W}{\mathrm{d}x}$.  As standard $p = - \mathrm{i} \:  \frac{\mathrm{d}}{\mathrm{d}x}$, i.e., we are using the  Schr\"{o}dinger representation. We will further impose the requirement that $|W| \rightarrow \infty$ as $x \rightarrow \pm \infty$. In this way we have only bound states, a discrete spectrum and critically, all the wave functions belong to the Hilbert space $\mathcal{H}$, i.e., we do not have plane wave solutions. Furthermore, for simplicity, we will always work with the full real line and not the half-line or some interval. This will avoid us having to discuss boundary conditions for different subspaces of the Hilbert space.
\begin{remark}
One can relax the smoothness condition on $W$ for just $C^1$, however for convenience we will insist on smoothness.
\end{remark}
\begin{theorem}
The operators defined above in \eqref{eqn:Model} satisfy the $\mathbb{Z}_2^2$-supertranslation algebra \eqref{eqn:Z22SuperAlg}. In other words, we have a $\mathbb{Z}_2^2$-SUSY QM system (see Definition \ref{def:ZQMSystem}) defined by the Hilbert space  $\mathcal{H} :=  L^2(\mathbb{R}) \otimes \mathbb{C}^4$  and the above operators \eqref{eqn:Model}.
\end{theorem}
\begin{proof}
We prove the theorem via direct computations.
\begin{itemize}
\item First we consider the self-commutator of $Q_{01}$.
$$ [Q_{01}, Q_{01}]  = \frac{1}{2} \left(\frac{p^2}{2 m} \otimes \Sigma_1^2 + W^2 \otimes \Sigma_2^2 \right) + \frac{1}{2} \left( \frac{p}{\sqrt{2m}} W \otimes \Sigma_1 \Sigma_2  + W  \frac{p}{\sqrt{2m}}  \otimes \Sigma_2 \Sigma_1 \right).$$
Now we use $\Sigma_1^2 = \Sigma_2^2 = \Sigma_0$ and $[\Sigma_1, \Sigma_2] = 0$ (being careful with the $\mathbb{Z}_2^2$-degree) to obtain
$$ [Q_{01}, Q_{01}] =  \frac{1}{2} \left(\frac{p^2}{2 m} + W^2 \right) \otimes \Sigma_0 + \frac{1}{2} \left( \frac{- \mathrm{i} \hbar}{ \sqrt{2m}} W^\prime \otimes \Sigma_1 \Sigma_2 \right).$$
Next we use $- \mathrm{i} \Sigma_1 \Sigma_2 = \Sigma_3$ and we arrive at
$$ [Q_{01}, Q_{01}] =  \frac{1}{2} \left(\frac{p^2}{2 m} + W^2 \right) \otimes \Sigma_0 + \frac{1}{2} \left( \frac{\hbar}{ \sqrt{2m}} W^\prime \otimes \Sigma_3 \right),$$
as required.
\item Next we consider the self-commutator of $Q_{10}$.
 $$ [Q_{10}, Q_{10}]  = \frac{1}{2} \left(\frac{p^2}{2 m} \otimes \alpha_2^2 + W^2 \otimes \alpha_1^2 \right) - \frac{1}{4} \left( \frac{p}{\sqrt{2m}} W \otimes \alpha_2 \alpha_1  + W  \frac{p}{\sqrt{2m}}  \otimes \alpha_1 \alpha_2 \right).$$
Now we use $\alpha_1^2 = \alpha_2^2 = \Sigma_0$ and $[\alpha_1, \alpha_2] = 0$  to obtain
$$ [Q_{10}, Q_{10}] =  \frac{1}{2} \left(\frac{p^2}{2 m} + W^2 \right) \otimes \Sigma_0 + \frac{1}{2} \left( \frac{ \mathrm{i} \hbar}{ \sqrt{2m}} W^\prime \otimes \alpha_2 \alpha_1 \right)\,.$$
Next we use $ \mathrm{i} \alpha_2 \alpha_1 = \Sigma_3$ and we arrive at the required expression.
\item Next we need to consider the `mixed' commutator of the $\mathbb{Z}_2^2$-supercharges.
\begin{align}[Q_{10},Q_{01}] &= \frac{1}{4} \left(\frac{p^2}{2m} \otimes [\alpha_2, \Sigma_1] - W^2 \otimes [\alpha_1, \Sigma_2] + \frac{p}{\sqrt{2m}} W \otimes \big(\alpha_2 \Sigma_2  + \Sigma_1 \alpha_1 \big)\right.\nonumber\\ &-\left. W  \frac{p}{\sqrt{2m}} \otimes\big(\alpha_1 \Sigma_1  + \Sigma_2 \alpha_2 \big)  \right)\, .\end{align}
Now we use $$[\alpha_2 ,\Sigma_1] =  - 2\mathrm{i} \alpha_3,  [\alpha_1 ,\Sigma_2] = 2\mathrm{i} \alpha_3, (\alpha_2 \Sigma_2  + \Sigma_1 \alpha_1) = (\alpha_1 \Sigma_1  + \Sigma_2 \alpha_2) = 2 \alpha_0 $$ and obtain
$$ [Q_{10}, Q_{01}] =  - \frac{ \mathrm{i}}{2} \left(   \left(\frac{p^2}{2 m} + W^2 \right) \otimes \alpha_3 +  \frac{\hbar}{ \sqrt{2m}} W^\prime \otimes \alpha_0  \right)\,,$$
as required.
\item It is clear that $[Q,H] =0$ follows from the Jacobi identity for the commutators. Thus, the only thing that now requires checking is if that $Z_{11}$ is indeed central.  From the $\mathbb{Z}_2^2$-supersymmetry algebra and the Jacobi identity, we need only show that $[Q_{01}, Z_{11}] =0$ and  $[Q_{10}, Z_{11}] =0$. For notational ease, we define $H_0 := \frac{p^2}{2m} + W^2$.
\begin{align*}
[Q_{01}, Z_{11}]  &= - \frac{\mathrm{i}}{2} \left(  \frac{p}{\sqrt{2m}}H_0 \otimes \Sigma_1 \alpha_3 + \frac{\hbar}{2m} pW^\prime  \otimes \Sigma_1 \alpha_0 + W H_0 \otimes \Sigma_2 \alpha_3  + \frac{\hbar}{\sqrt{2m}} W W^\prime \otimes \Sigma_2 \alpha_0 \right.\\
 & \left. +  H_0 \frac{p}{\sqrt{2m}} \otimes \alpha_3 \Sigma_2  + H_0 W \otimes \alpha_3 \Sigma_2 + \frac{\hbar }{2m} W^\prime p \otimes \alpha_0 \Sigma_1  + \frac{\hbar}{\sqrt{2m}} W^\prime W \otimes \alpha_0 \Sigma_2\right) \\
 &=  - \frac{\mathrm{i}}{2}\left( \frac{1}{\sqrt{2m}}[p,W^2] \otimes \Sigma_1 \alpha_3  + \frac{2\hbar}{\sqrt{2m}} W W^\prime \otimes \Sigma_2 \alpha_0\right.\\ &+ \left. \frac{1}{2m} [W , p^2] \otimes \Sigma_2 \alpha_3 + \frac{\hbar}{2m} (p W^\prime + W^\prime p) \otimes \Sigma_1 \alpha_0\right) \\
  &= - \frac{\mathrm{i}}{2} \left( \left(\frac{-\mathrm{i}}{\sqrt{2m}}[p, W^2] + \frac{2 \hbar}{\sqrt{2m}} W W^\prime\right)\otimes \alpha_2\right.\nonumber\\&   + \left(\frac{\mathrm{i}}{2m} [W,p^2] + \frac{\hbar}{2m}\big( pW^\prime + W^\prime p \big) \otimes \alpha_1\right) .
\end{align*}
Evaluation of the commutators shows that the above expression vanishes.
\bigskip

A similar computation  gives
\begin{align}[Q_{10}, Z_{11}]  = - \frac{\mathrm{i}}{2}  \left(\frac{\mathrm{i}}{\sqrt{2m}} [p,W^2] - \frac{2\hbar }{\sqrt{2m}}W W^\prime\right) \otimes \Sigma_1   \nonumber \\ -\frac{\mathrm{i}}{2}\left( \frac{\hbar }{2m}\big(pW^\prime + W^\prime p \big) + \frac{\mathrm{i}}{2m} [W,p^2]  \right)\otimes \Sigma_2 \,.\end{align}
Then once again by evaluating the commutators this expression vanishes.
 \end{itemize}
\end{proof}
\noindent We will make the canonical identification between a matrix  $O$ and $(1\otimes O)$ as needed.
\begin{proposition}
The Hamiltonian and central charge are related by
$$H_{00} \: \alpha_3 = {-} \mathrm{i} \: Z_{11}.$$
\end{proposition}
\begin{proof}
This follows as a simple consequence of the fact that $\Sigma_0 \alpha_3 = \alpha_3$  and $\Sigma_3 \alpha_3 = \alpha_0$.
\end{proof}
  \subsection{Ladder-like Operators}
  We define the two following operators
  \begin{align*}
  &A := \frac{ \mathrm{i}p}{\sqrt{2m}} +  W, && A^\dag := {-}\frac{\mathrm{i}p}{\sqrt{2m}}+ W\,,
  \end{align*}
which are identical to those found in Witten's model. We then write the following:
  \[ Q_{01} = \frac{\mathrm{i}}{2}\left(\left( \begin{array}{rr}
0 & {-}A  \\
A^\dag & 0  \\
\end{array} \right)  \oplus  \left( \begin{array}{rr}
0 & {-}A  \\
A^\dag & 0  \\
\end{array} \right) \right) =  \frac{\mathrm{i}}{2}\left( \begin{array}{rrrr}
0 & {-}A & 0 & 0 \\
A^\dag & 0 & 0 & 0  \\
0 & 0 & 0 & {-}A\\
0 & 0 & A^\dag & 0\\
\end{array} \right), \]
\[
 Q_{10} =  {-}\frac{\mathrm{1}}{2}\left(\left( \begin{array}{rr}
0 & A  \\
A^\dag & 0  \\
\end{array} \right)  \ominus \left( \begin{array}{rr}
0 & A  \\
A^\dag & 0  \\
\end{array} \right) \right) =   {-}\frac{\mathrm{1}}{2}\left( \begin{array}{rrrr}
0 & 0 & 0 & A \\
0 & 0 & A^\dag & 0  \\
0 & A & 0  & 0\\
A^\dag & 0 & 0  & 0\\
\end{array} \right). \]

\bigskip

  Let us define $H_+ :=  \frac{p^2}{2m} + W^2 + \frac{\hbar}{\sqrt{2m}} W^\prime$ and $H_- :=  \frac{p^2}{2m} + W^2 - \frac{\hbar}{\sqrt{2m}} W^\prime$. A quick computation shows that $H_+ = A A^\dag$ and $H_-= A^\dag A$. Using this, we can obtain:

   \[ H_{00} = \left( \begin{array}{rrrr}
H_+ & 0 & 0 & 0 \\
0 & H_- & 0 & 0  \\
0 & 0 & H_+ &0\\
0 & 0 & 0 & H_-\\
\end{array} \right), \hspace{25pt}
 Z_{11} = - \mathrm{i} \left( \begin{array}{rrrr}
0 & 0 & H_+ & 0 \\
0 & 0 &0 & - H_-  \\
H_+ & 0 & 0  & 0\\
0 & - H_- & 0  & 0\\
\end{array} \right). \]

\subsection{Parity Operators}
The Hilbert space we are considering has a natural decomposition
$$\mathcal{H} = \mathcal{H}_{00} \oplus \mathcal{H}_{01} \oplus \mathcal{H}_{11} \oplus \mathcal{H}_{10}\,.$$
As we have a bi-grading, naturally  have a pair of  parity operators (by definition these must be degree $(0,0)$ operators). Explicitly,
  \[ \mathbb{K}_1 = \sigma_0 \oplus (- \sigma_0) = \left( \begin{array}{rrrr}
1 & 0 & 0 & 0 \\
0 & 1 & 0 & 0  \\
0 & 0 & {-}1 &0\\
0 & 0 & 0 & {-}1\\
\end{array} \right), \hspace{25pt}
 \mathbb{K}_2 = \sigma_3 \oplus \sigma_3 = \left( \begin{array}{rrrr}
1 & 0 & 0 & 0 \\
0 & {-}1 & 0 & 0  \\
0 & 0 & {-1} &0\\
0 & 0 & 0 & 1\\
\end{array} \right). \]
Direct computation establishes the following:
\begin{proposition}
The pair of  parity operators $\mathbb{K}_1$ and $\mathbb{K}_2$ satisfy the following relations.
\begin{align}
&[\mathbb{K}_1, \mathbb{K}_2] =0,\\
& (\mathbb{K}_1)^2 = (\mathbb{K}_2)^2 = \Sigma_0.
\end{align}
Moreover,
\begin{equation}
[\mathbb{K}_1, H_{00}] = [\mathbb{K}_2, H_{00}] =0.
\end{equation}
\end{proposition}
\begin{corollary}\label{cor:SimEigen}
Simultaneous eigenstates of the Hamiltonian $H_{00}$ and the parity operators $\mathbb{K}_1$ and $\mathbb{K}_2$ exist.
\end{corollary}
The above result will be essential in describing explicitly how the $\mathbb{Z}_2^2$-supercharges act on energy eigenstates. \par
From direct calculation we observe (rather naturally) that
\begin{align}
& \mathbb{K}_1 Q_{01} = + \;Q_{01} \mathbb{K}_1, && \mathbb{K}_2 Q_{01} = - \;Q_{01} \mathbb{K}_2,\\
& \mathbb{K}_1 Q_{10} = - \;Q_{01} \mathbb{K}_1, && \mathbb{K}_2 Q_{10} = + \;Q_{10} \mathbb{K}_2,\\
& \mathbb{K}_1 Z_{11} = - \;Z_{11} \mathbb{K}_1, && \mathbb{K}_2 Z_{11} = - \;Z_{11} \mathbb{K}_2.
\end{align}
\begin{remark}
In the definition of a $\mathbb{Z}_2^2$-SUSY QM system (see Definition \ref{def:ZQMSystem}) we do \underline{not} postulate the existence of the parity operators that commute with the Hamiltonian. However, the existence of such operators is essential in describing $\mathbb{Z}_2^2$-supersymmetry as a mapping  between states of different $\mathbb{Z}_2^2$-degrees.
\end{remark}
\subsection{Energy eigenstates and $\mathbb{Z}_2^2$-supersymmetry}
It is clear that
\begin{align}
& Q_{01}\mathcal{H}_{ij} \subset \mathcal{H}_{i(j+1)}, && Q_{10}\mathcal{H}_{ij} \subset \mathcal{H}_{(i+1)j}\,.
\end{align}
Due to Corollary \ref{cor:SimEigen}, we can present the action of the $\mathbb{Z}_2^2$-supercharges on energy eigenstates rather explicitly and construct the corresponding multiplets. In particular, we can label the simultaneous eigenstates of  $H_{00}$, $\mathbb{K}_1$ and $\mathbb{K}_2$ as
$$| E, i,j \rangle  \in \mathcal{H}_{ij}\,.$$
\begin{proposition}\label{prop:EnergyEigenValues}
For every state $|E, i,j\rangle$, with energy eigenvalue $E>0$, there exists two other states $|E, i , j+1\rangle$ and $|E, i+1,j\rangle$, with the same energy eigenvalue $E$.
\end{proposition}
\begin{proof}
We define $|E, i,j+1\rangle  :=  \frac{2}{\sqrt{E}} Q_{01}|E, i,j\rangle$. Now we must check that this really is an energy eigenstate with eigenvalue $E$. This follows as $[H,Q] =0$. Explicitly
\addtolength{\jot}{-1em}
\begin{align*}
 &H_{00}|E ,i,j+1\rangle  = H_{01} \left(\frac{2}{\sqrt{E}} Q_{01}|E, i,j\rangle\right)
  \\[9pt]&=  \frac{2}{\sqrt{E}} Q_{01} \left(  H_{00}|E,i,j\rangle\right)
  = \frac{2}{\sqrt{E}} Q_{01} \left( E|E, i,j\rangle\right) = E\,  |E,i,j+1\rangle.
\end{align*}
Similarly, we define $|E, i+1, j\rangle  :=  \frac{2}{\sqrt{E}} Q_{10}|E, i,j\rangle$ and an almost identical calculation to the above shows that this state is also an energy eigenstate of energy $E$.
\end{proof}
\begin{remark}
Proposition \ref{prop:EnergyEigenValues} does \underline{not} depend on the specific model, but \underline{only} the existence of the pair of parity operators that commute with the Hamiltonian.
\end{remark}
Now let us proceed to the specifics of this model and define the corresponding  wave function in the position representation
$$\Phi^E_{ij}(x) := \langle x|  E, i,j \rangle.$$
In particular,
\def\arraystretch{1.4}
\addtolength{\jot}{1em}
$$\Phi^E_{00}(x)  =
\left( \begin{array}{c}
\psi^E_{00}(x) \\
0  \\
0 \\
0 \\
\end{array} \right),
$$
where $\psi^E_{00}(x) \in L^2(\mathbb{R})$.  Clearly it is an eigenvector of $H_{+}$, i.e., $H_+ \psi^E_{00} = E \,\psi^E_{00}$. Similar expression hold for the other degree components.  We again assume the energy to be strictly greater than zero.  We define the $\mathbb{Z}_2^2$-multiplet generated by $\Phi^E_{00}$ as all the states that can be `reached' by application of $Q_{01}$, $Q_{10}$. We will normalise the operators  for convenience and define
\begin{align*}
\bar{Q}_{01} := \frac{2}{\sqrt{E}}Q_{01}, && \bar{Q}_{10} := \frac{2}{\sqrt{E}}Q_{10}.
\end{align*}
Direct computation gives
 \begin{align}
 & \bar{Q}_{01}  \Phi^E_{00}(x)  =
\left( \begin{array}{c}
0 \\
\frac{\mathrm{i}}{\sqrt{E}}A^\dag \psi^E_{00}(x)  \\
0 \\
0 \\
\end{array} \right), && \bar{Q}_{10}  \Phi^E_{00}(x)  =
\left( \begin{array}{c}
0 \\
0  \\
0 \\
- \frac{1}{\sqrt{E}} A^\dag \psi^E_{00}(x) \\
\end{array} \right),\\
&\bar{Q}_{10} \bar{Q}_{01}  \Phi^E_{00}(x)  =
\left( \begin{array}{c}
0 \\
0   \\
-\mathrm{i}\psi^E_{00}(x) \\
0 \\
\end{array} \right), && \bar{Q}_{01}\bar{Q}_{10}  \Phi^E_{00}(x)  =
\left( \begin{array}{c}
0 \\
0  \\
\mathrm{i}\psi^E_{00}(x)\\
0 \\
\end{array} \right).
\end{align}
Note that the last two states are clearly not linearly independent. Thus, the
 $\mathbb{Z}_2^2$-multiplet generated by $\Phi^E_{00}$ is the sub-vector space of $\mathcal{H}$ spanned by the following elements:
 \begin{align}
 & \left( \begin{array}{c}
\psi^E_{00}(x)\\
0 \\
0 \\
0 \\
\end{array} \right), &&
\left( \begin{array}{c}
0 \\
A^\dag \psi^E_{00}(x)  \\
0 \\
0 \\
\end{array} \right),
&
\left( \begin{array}{c}
0 \\
0  \\
0 \\
 A^\dag \psi^E_{00}(x) \\
\end{array} \right),
&&
\left( \begin{array}{c}
0 \\
0   \\
\psi^E_{00}(x) \\
0 \\
\end{array} \right).
\end{align}
\begin{remark}
Similar conclusions can be drawn by looking at $\mathbb{Z}_2^2$-multiplets generated by another reference states.
\end{remark}
Clearly, we have the following.
\begin{proposition}
Energy levels with $E>0$ are four-fold degenerate.
\end{proposition}

\subsection{Zero energy states}
To examine the nature of the zero energy states, let us consider the following states (written in the position representation),
\begin{align}
 & \left( \begin{array}{c}
\psi_{00}(x)\\
0 \\
0 \\
0 \\
\end{array} \right), &&
\left( \begin{array}{c}
0 \\
\psi_{01}(x)  \\
0 \\
0 \\
\end{array} \right),
&
\left( \begin{array}{c}
0 \\
0  \\
\psi_{11}(x) \\
 0 \\
\end{array} \right),
&&
\left( \begin{array}{c}
0 \\
0   \\
0 \\
\psi_{10}(x) \\
\end{array} \right).
\end{align}

\bigskip

Each component belongs to $L^2(\mathbb{R})$. Applying $Q_{01}$ and separately $Q_{10}$ to these states and insisting that the result vanishes gives the conditions
\begin{align}
& A^\dag \psi_{00} =0, && A^\dag \psi_{11} =0,\\
& A \psi_{01} =0 , && A \psi_{10} =0.
\end{align}
 We know from the analysis of  Witten's model \cite{Wit} that all these conditions cannot be consistent for non-trivial wave functions - some of them or all of these wave functions must be zero. Thus, the zero energy states (in the position representation) are {either} of the form
 \begin{equation}
X(x) = a \left( \begin{array}{c}
\chi(x)\\
0 \\
0 \\
0 \\
\end{array} \right)
+
b \left( \begin{array}{c}
0\\
0 \\
\chi(x) \\
0 \\
\end{array} \right) \in \mathcal{H}_{00}\oplus \mathcal{H}_{11},
 \end{equation}
 or
  \begin{equation}
\Psi(x) = c  \left( \begin{array}{c}
0\\
\psi(x) \\
0 \\
0 \\
\end{array} \right)
+
d \left( \begin{array}{c}
0\\
0 \\
0 \\
\psi(x)\\
\end{array} \right) \in \mathcal{H}_{01}\oplus \mathcal{H}_{10},
 \end{equation}
 where $A^\dag \chi =0$, $A \psi =0$ and $a,d,c$ and $d \in \mathbb{C}$.

\bigskip

We have thus proved the following.
\begin{proposition}
A zero energy state is either non-existent or two-fold degenerate and belongs to either  $\mathcal{H}_{00}\oplus \mathcal{H}_{11}$ or $\mathcal{H}_{01}\oplus \mathcal{H}_{10}$, i.e., to either the even or odd subspaces of $\mathcal{H}$.
\end{proposition}

\subsection{The Harmonic Oscillator Potential}
Note that the spectrum our $\mathbb{Z}_2^2$-SUSY QM model is identical to that of the corresponding Witten model. The key difference is in the degeneracy of the energy eigenstates: everything is doubled.  In order to illustrate the structure of our  model  explicitly we will examine the potential
\begin{equation}\label{eqn:HarPot}
W(x) =  {-} \sqrt{\frac{m}{2}} \omega x\,
\end{equation}
where $\omega >0$. Clearly, we have
\begin{align}
H_+ = \frac{p^2}{2m} + \frac{m \omega^2}{2}x^2 - \frac{\hbar}{2}\omega, && H_- = \frac{p^2}{2m} + \frac{m \omega^2}{2}x^2 + \frac{\hbar}{2}\omega\,.
 \end{align}
 We have chosen the sign in the potential so that the zero energy states belong to $\mathcal{H}_{00} \oplus \mathcal{H}_{11}$. This is just for our convenience. We see that the  Hamiltonians are just the standard Hamiltonians for the harmonic oscillator with constant shifts by one unit of the energy. Thus, the wave functions are just the usual wave functions for the harmonic oscillator. As standard we define
$$\psi_n(x) = \frac{1}{\sqrt{2^n n!}} \left(\frac{m \omega}{\pi \hbar} \right)^{\frac{1}{4}} \exp \left(- \frac{m w}{2 \hbar} x^2 \right)~ H_{n}\left( \sqrt{\frac{m \omega}{\hbar}} x\right)\,,$$
were $H_n$ is the $n$-th Hermite polynomial.  It is straightforward to see that the spectrum of $H_{00}$ with the potential \eqref{eqn:HarPot} is
$$0, ~\hbar \omega, ~2 \hbar \omega, ~3 \hbar \omega, ~\cdots .$$
This is of course, identical to the Witten model. The zero energy ground states are (weighted) sums of the vectors
\begin{align}
&  \left( \begin{array}{c}
\psi_0(x)\\
0 \\
0 \\
0 \\
\end{array} \right), &&
\left( \begin{array}{c}
0\\
0 \\
\psi_0(x) \\
0 \\
\end{array} \right)\,.
\end{align}
 The excited states $(n\geq 1)$ are then (weighted) sums of the vectors
 \begin{align}
&  \left( \begin{array}{c}
\psi_n(x)\\
0 \\
0 \\
0 \\
\end{array} \right), &&
\left( \begin{array}{c}
0\\
\psi_{n-1}(x) \\
0 \\
0 \\
\end{array} \right)\,,
\nonumber&  \left( \begin{array}{c}
0\\
0 \\
\psi_n(x)\\
0 \\
\end{array} \right), &&
\left( \begin{array}{c}
0\\
0 \\
0 \\
\psi_{n-1}(x) \\
\end{array} \right)\,.
\end{align}
We see the  two-fold degeneracy in a zero energy ground state and the four-fold degeneracy for the excited states.

\section{\scshape Conclusions}\label{sec:Concl}
Thus, we have constructed a double-graded supersymmetric quantum mechanical model in order to illustrate that employing ``higher gradings'' in physics, and in particular, quantum mechanics is possible and can lead to interesting results. Specifically, we have a realization of the $\mathbb{Z}_2^2$-supertranslation algebra as proposed in \cite{Bru} via a quantum mechanical model. This model is not equivalent to the models with para-Grassmann variables \cite{ger/tka84,ger/tka85}
and parasupersymmetries \cite{Bec/Deb2,Rub/Spi} , nor with various versions of the
(extended) supersymmetric quantum mechanics
\cite{aku/ceb/pas,aku/dup2,aku/pas,fub/rab}. Indeed, the direct physical interpretation of   $\mathbb{Z}_2^2$-SUSY QM is currently lacking. We expect more involved models of   $\mathbb{Z}_2^n$-supersymmetric quantum mechanics, i.e., ``multiple-graded supersymmetric quantum mechanics'' to be important and lead to more unexpected results.

\bigskip

\textbf{Acknowledgements}.
\addcontentsline{toc}{section}{Acknowledgements}
The second author (S.D.) is grateful to V. Akulov and V. Tkach for fruitful discussions, to V. N. Tolstoy for sending his papers, and is very thankful to N. Poncin for numerous conversations and his  kind hospitality at University of Luxembourg, where this paper was initiated and started in December, 2018.

\newpage

\mbox{}

\end{document}